\documentclass[letterpaper,11pt]{article}
\usepackage{chicagor}
\usepackage{verbatim}
\usepackage{xspace}
\usepackage{amsmath}
\usepackage{amsthm}
\usepackage{amssymb}
\usepackage{graphicx}
\usepackage{epstopdf}
\usepackage{url}
\newcommand{\commentout}[1]{}
\newcommand{\LA}{\mathit{LA}}

\title{Computational Extensive-Form Games}
\author{
Joseph Y. Halpern \qquad \qquad Rafael Pass \qquad \qquad  Lior Seeman\\ 
Computer Science Dept.\\
Cornell University\\
Ithaca, NY\\
E-mail: halpern$|$rafael$|$lseeman@cs.cornell.edu
}

\newtheorem{theorem}{Theorem}[section]

\newtheorem{lemma}[theorem]{Lemma}

\newtheorem{definition}[theorem]{Definition}

\newtheorem{exm}[theorem]{Example}
\newtheorem*{theorem*}{Theorem}
\newtheorem*{corollary*}{Corollary}
\newtheorem*{conjecture*}{Conjecture}
\newtheorem*{lemma*}{Lemma}
\newtheorem*{thm*}{Theorem}
\newtheorem*{prop*}{Proposition}
\newtheorem*{obs*}{Observation}
\newtheorem*{rem*}{Remark}
\newtheorem*{definition*}{Definition}
\newtheorem*{exm*}{Example}

\newcommand{\G}{\mathcal{G}}
\newcommand{\F}{\mathcal{F}}
\newcommand{\I}{\mathcal{I}}
\newcommand{\<}{\langle}
\renewcommand{\>}{\rangle}

\newif\ifstateless

\begin{document}

\maketitle

\begin{abstract}
We define solution concepts appropriate for computationally bounded
players playing a fixed finite game.  To do so, we need to define what it
means for a \emph{computational game}, which is a 
sequence of games that get larger in some appropriate sense, to
represent a single finite underlying 
extensive-form game.  
Roughly speaking, we require all the games in the sequence to have essentially
the same structure as the underlying game, except that 
two histories that are indistinguishable (i.e., in the same
information set) in the underlying game may correspond to histories
that are only computationally
indistinguishable in the computational game.  
We define a computational version of both Nash equilibrium
and sequential equilibrium for computational games, and show that
every Nash (resp., sequential) equilibrium in the underlying game corresponds
to a computational Nash (resp., sequential) equilibrium in the computational 
game.  
One advantage of our approach is that 
if a cryptographic protocol represents an abstract game,
then we can analyze its strategic behavior in the abstract game, and
thus separate the cryptographic analysis of the protocol from the
strategic analysis.  
\end{abstract}

\section{Introduction}

Game-theoretic models assume that the players are completely rational.
This is typically interpreted as saying that payers act optimally 
given (their beliefs about) other players' behavior. However, 
as was first pointed out by Simon~\citeyear{Simon55}, acting
optimally may be hard.  Thus, there has been a great deal of
interest in capturing \emph{bounded rationality}, and finding solution
concepts appropriate for resource-bounded players.

One explanation of bounded rationality is that players have
limits on their computational power.
For example, the players might be able to use only strategies that can
be implemented by a polynomial-time TM. 
While there has been a great deal of work%
~\cite{DHR00,GLR13,HPS14,HPS14b,hubavcek2013limits,hubavcek2014cryptographically,UV04} 
on solving game-theoretic problems using computationally bounded players, 
there has not really been a careful study of the solution concepts
appropriate for such players.  What does it mean, for
example, to say that a fixed finite game played by polynomial-time
players has a Nash equilibrium?   

Consider for example the following two-player extensive-form game $G$:  At
the the empty history, player $1$ secretly chooses one of two
alternatives and puts her choice inside a sealed envelope.  
Player $2$ then also chooses one of these
two alternatives.  
Finally, player $1$ can either open the envelope and
reveal her choice or destroy the envelope.  
If she opens the envelope and she chose a different alternative than
player $2$, player 1 wins and gets a utility of 1; otherwise (i.e., if
player 1 either chose the same alternative as player $2$ or she 
destroyed the envelope) player 1 loses and gets a utility of $-1$.   
Player's 2's utility is the opposite of player 1's.
The game tree for this game is given in Figure~\ref{fig:extensiveGame}. 
Since player 2 acts without knowing
1's choice, the two histories where 1 made different choices are in
the same information 
set of player 2. 

\begin{figure}
                \centering
                \includegraphics[width=0.5\textwidth]{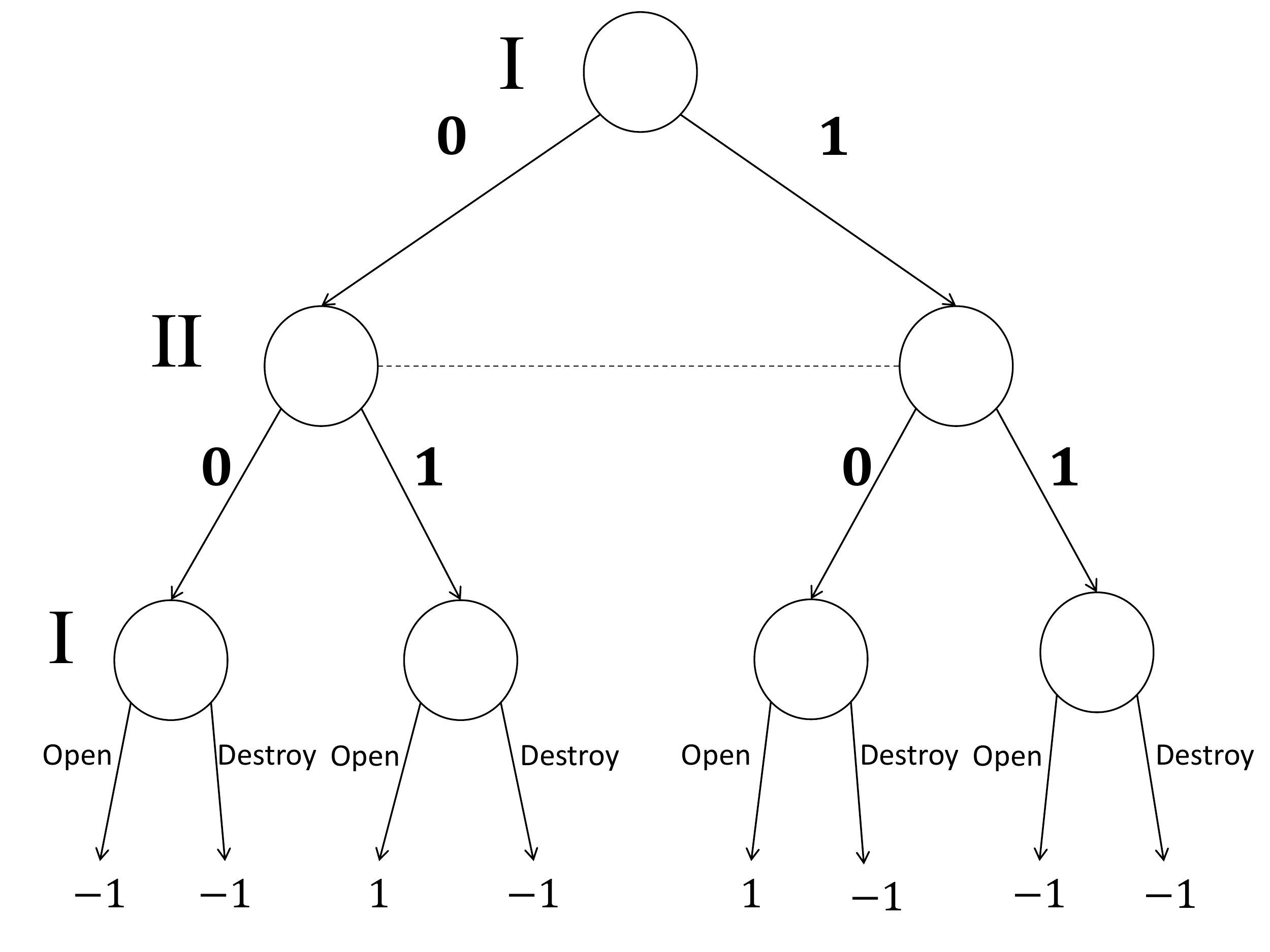} 
 	\caption{A game that can be represented by a computational game.}
	\label{fig:extensiveGame}
\end{figure}


Resource-bounded players can implement this game even without access
to envelopes, using what is called a \emph{commitment scheme}.
A commitment scheme is a two-phase two-party
protocol involving a sender (player 1 above) and a receiver (player
2).  The sender sends the receiver a message in the first phase that
commits him to a bit without giving the receiver information
about the bit (at least no information that he can efficiently compute from
the message); this is the computational analogue of putting the
bit in an envelope.  In the second phase, the sender ``opens the
envelope'' by sending the receiver some information that allows the receiver to
confirm what bit the sender committed to in the first phase.
Thus, we can talk about a game $\G$ (actually a sequence of games as
discussed later) where instead of player 1 using an abstract envelope
to send her choice to player 2, she uses a commitment scheme to do
so. 

Intuitively, we would like to say that the two games represent the
same underlying game. However, there are many subtleties in doing so. 
To get a sense of the problems, note
that to use commitment schemes we need the players to be
computationally bounded. But to 
talk about computation bounds (for instance, polynomial-time TMs), we
need to have a sequence of inputs  
that can grow as a function of $n$.  
So how do we proceed if we want to
talk about computationally bounded players in a fixed finite game?
The idea is that we will have a sequence of games, potentially
increasing in size, that represents the single game. 
As we shall see, the information structure of the games in the
infinite sequence might differ from that of the underlying game.
For example, in the games described before, while
a commitment scheme gives no information to a computationally bounded
player, an unbounded player has  complete information; the encrypted
string uniquely identifies the bit that was committed. 
Thus, unlike in $G$, 
commitments to different bits in $\G$ are in different information
sets for player $2$. 

Additional complications arise when we consider solution concepts for
such games.  
Traditional notions of equilibrium
involve all players
making a best response.  But if we restrict to computationally bounded
players, there may not be a best response, especially for the kinds of
cryptographic problems that we would like to consider.  For example, for every
polynomial-time TM, there may be another TM that does a little better
by spending a little longer trying to do decryption.
(See \cite{HPR15} for an example of this phenomenon.)
Moreover, when considering sequentially rational solution concepts it
is unclear what information structure should be considered since, as we
discussed, the information structure of the computational games does
not capture the knowledge of computationally bounded players.

\paragraph{Our contributions.}
As a first step to capturing these notions, 
in Section~\ref{sec:definitions},
we define 
what it means
for a sequence $\G = (G_1, G_2, \ldots)$ of games to represent a
single game $G$.  Intuitively, all the games in the sequence $\G$
have the same basic structure as $G$, but
might use increasingly longer strings to represent actions in $G$ (e.g.,
an action $a$ in $G$ might be represented in $G_n$ by an encryption of
$a$ that uses a security parameter of length $n$).
More precisely, we require a mapping from histories in the games $G_n$
to histories in $G$,
as well as a mapping from strategies in $G$ to strategies in $\G$,
and impose what we argue are reasonable
conditions on these mappings.%
\footnote{The idea of games that depends on a
  security parameter goes back to Dodis, Halevi, and
Rabin~\citeyear{DHR00}.
Hub{\'a}{\v{c}}ek and
  Park~\citeyear{hubavcek2014cryptographically} also consider a mapping
  between histories in a  
  computational game and histories in an abstract game, 
  although they do not
  consider the questions in the same generality that we do here.}
In Section~\ref{sec:commitGame}, we show how this definition play out
in the example discussed above. 

As hinted before, our conditions do not force the games in $\G$ to
have the same information
structure as $G$. While two histories in the same
information set in $G_n$ must map to two histories in the same
information set in $G$, it may also be the case
that two histories in different information sets in $G_n$ are 
mapped to the same information set in $G$.%
\commentout{
To understand why we want
to allow this, suppose that in $G_n$, there are histories $h_1$ and
$h_2$ where a bit 0 in encrypted using two different keys by agent 1.
Agent 2 can distinguish $h_1$ and $h_2$, because the
encryptions are two different strings; thus, $h_1$ and $h_2$ are in different
information sets for agent 2.  And they are both in a different
information set from $h_3$, where a bit 1 is encrypted. 
Nevertheless, both $h_1$ and $h_2$ are
mapped to the same history $h$ in $G$, where an unencrypted bit 0 is
put in an envelope, while $h_3$ is mapped to a history $h'$, where an
unencrypted bit 1 is put in an envelope, which is in the same
information set as $h$.  
(See the example in Section~\ref{sec:commitGame} for more intuition.)
}
Although a player can distinguish two histories in different
information sets (for example a
commitment to $0$ and a commitment to $1$ in the example are two
different strings), 
at a computational level, she cannot tell them apart.  The encodings
just look like random strings to her.  There is a sense in which she,
as a computationally bounded player,
does not understand the ``meaning'' of these histories (although a
computationally unbounded player could break the commitment and tell
them apart).  
In Section~\ref{sec:compInfSet}, we make this intuition precise, 
showing that our requirements force all histories that map to the same
information set in $G$ to be \emph{computationally indistinguishable}, even if they are in different information sets
in $\G$.%


Once we have defined our model of computational games, 
we focus on defining analogues of two solution concepts, Nash
equilibrium (NE) and sequential equilibrium. In
Section~\ref{sec:compNash}, we define a computational analogues of NE,
which considers only deviations that can be implemented by
polynomial-time TMs. It handles previously mentioned complications by
allowing for the strategy to be an $\epsilon$ best response for some
negligible function $\epsilon$.  (Our definition of NE is similar in spirit
to the definition in Dodis, Halevi, and 
Rabin~\citeyear{DHR00}.)
We show that if a strategy profile is a NE in the underlying game $G$, then there is a corresponding strategy profile of polynomial time TMs that is a computational NE in
$\G$. 
Thus, we provide conditions that guarantee the existence of a
computational NE, addressing an open question of Katz~\citeyear{Katz08}. 

In Section~\ref{sec:compSeq}, we define a computational analogue of
sequential equilibrium. 
It is notoriously problematic to define sequentially rational solution
concepts in cryptographic protocols. For example,  Gradwohl, Livne,
and Rosen~\citeyear{GLR13} 
provide a general discussion of the issue, and give a partial solution
in terms of avoiding what they call ``empty 
threats'', which applies only to two-player games of
perfect information, and discuss possible extensions.  Our notion of
computational sequential equilibrium, which is quite different in
spirit from the solution concept of Gradwohl, Livne, and Rosen (and
arguably conceptually much simpler and much closer in spirit to the
standard game-theoretic definition), 
applies to arbitrary sequence of games that represent a finite game,
and uses the intuitions we develop on the connection between
information sets in the underlying finite game and computational
indistinguishability in the sequence. 
We again show that if a strategy profile is a sequential equilibrium in the underlying game $G$, then there is a corresponding strategy profile of polynomial time TMs that is a computational sequential equilibrium in $\G$.


An important benefit of our approach is that it separates the game-theoretic
analysis from the cryptographic analysis. 
We can view the sequence $\G$ as an implementation of an abstract
game $G$. 
\commentout{
Under this view, the relationship between $\G$ and $G$ is
similar in spirit to the relationship
between ideal and real worlds often used in describing cryptographic
protocols. We can view the ideal protocol as an abstract game $G$ and
the sequence $\G$ as implementation of it, using increasing security
parameters.
}
Given this view, we can first prove that a protocol is a good
implementation of an abstract game, and then analyze the strategic
aspects in that simple abstract game.  
For example,
to show a prescribed cryptographic protocol is a Nash (resp.,
sequential) equilibrium, we can first show it represents
an abstract ideal game; it then 
suffices to show that
the protocol corresponds to a strategy profile that is a Nash (resp.,
sequential) equilibrium in the much simpler underlying game.
We give an exmaple of this idea in Section~\ref{sec:correlated}, where
we show how our approach can be 
used to analyze a protocol for implementing a correlated equilibrium
(CE) without a mediator using cryptography, in the spirit of the work
of Dodis, Halevy, and Rabin \citeyear{DHR00}.

\section{Preliminaries}\label{sec:review}

\subsection{Extensive-form games}
We begin by reviewing the formal definition of an extensive-form game \cite{OR94}. A
finite extensive-form game $G$ is a tuple $([c], H, P, \vec{u},
\vec{\I})$ satisfying the following conditions:
\begin{itemize}
\item $[c]=\{1,\ldots,c\}$ is the set of players in the game. 
\item $H$ is a set of history sequences that satisfies the following
two properties: 
\begin{itemize}
\item the empty sequence $\<\,\>$ is a member of $H$;
\item if $\<a_1, \ldots, a_K\>\in H$ 
and
$L<K$ then $\<a_1, \ldots, a_L\> \in H$. 
The elements of a history $h$ are called \emph{actions}. 
\end{itemize}
A history  $\<a_1, \ldots, a^K\>\in H$ is \emph{terminal} 
if there is no $a$ such that $\<a^1, \ldots, a^K, a\> \in
H$. The set of actions available after a nonterminal history $h$ is
denoted $A(h)=\{a:h\cdot a)\in H\}$ (where $h \cdots a$ is the result
of concatenating $a$ to the end of $h$.\footnote{For technical
  convenience, we assume that $|A(h)|\geq 2$ for
all histories $h$. If this is not the case, then that step of the game is not
interesting, and can essentially be removed.} 
Let $H^T$ denote the set of terminal histories, let $H^{NT}$ denote
$H\setminus H^T$, and let  
$H^i$ denote the histories after which player $i$ plays.
\item $P: H\setminus H^T\to [c]$. $P(h)$ specifies the player
that moves at history $h$. 
\item $\vec{u}: H_{T}\to\mathbb{R}^c$ specifies for each terminal
history the utility of the players at that history ($u_i(h)$ is 
the utility of player $i$ at terminal history $h$). 
\item for each player $i\in [c]$, $\mathcal{I}_i$ is a partition of $H^i$
with the property that $A(h)=A(h')$ whenever $h$ and $h'$ are in the
same member of the partition. For $I \in\mathcal{I}_i$, we denote by
$A(I)$ the set $A(h)$ for $h\in I$ (recall that $A(h) = A(h')$ if $h$  
and $h'$ are two histories in $I$). 
We assume without loss of generality that if $I \ne I'$, then $A(I)$
and $A(I')$ are disjoint (we can always rename actions to ensure that
this is the case).
We call $\mathcal{I}_i$ the 
\emph{information partition} of player $i$; a set $I\in\mathcal{I}_i$
is an \emph{information set} of player $i$; 
$\vec{\I}=(\I_1,\ldots,\I_c)$ is the \emph{information partition
structure} of the game.  
A game of \emph{perfect information} is one where all the information sets are
singletons.  
\end{itemize}

This model can capture situations in which players forget what they knew
earlier. Roughly speaking, a game has \emph{perfect recall} if the
information  
structure is such that the players remember everything they knew in the
past.
\begin{definition}
 Let $\mathit{EXP}_i(h)$ be the record of player $i$'s
 \emph{experience} in history 
$h$, that is, all the actions he plays and all the information sets he
encounters in the history. 
A game has \emph{perfect recall} if, for each player $i$, we have 
$\mathit{EXP}_i(h)=\mathit{EXP}_i(h')$ whenever the histories $h$ and
$h'$ are in the same 
information set for player $i$. 
\end{definition}

A deterministic strategy $s$ for player $i$ is a function from $\I_i$
to actions, where for $I \in \I_i$, we require that $s(I) \in A(I)$.
\ifstateless
\else
We also consider mixed strategies which are probability distribution over deterministic strategies.
A profile of strategies
\fi
$\vec{\sigma}=\{\sigma_1,\ldots,\sigma_c\}$ 
induces a distribution denoted $\rho_{\vec{\sigma}}$ on terminal histories.
We say that a strategy profile is completely mixed if $\rho_{\vec{\sigma}}$
assigns positive probability to every history $h\in H^T$. 
The expected value of player $i$ given $\vec{\sigma}$ is
then $\sum_{h\in H^T}\rho_{\vec{\sigma}}(h)u_i(h)$. 

We use the standard notation $\vec{x}_{-i}$ to denote the vector $\vec{x}$ with its $i$th
element removed and $(x',\vec{x}_{-i})$ to denote $\vec{x}$ with its
$i$th element replaced by $x'$.

\begin{definition}[Nash Equilibrium]
$\vec{\sigma}=\{\sigma_1,\ldots,\sigma_c\}$  is an \emph{$\epsilon$-Nash
equilibrium} (NE)  of $G$
if, for all players $i\in[c]$ and for all strategies $\sigma'$ for
player $i$, 
$$\sum_{h\in H^T}\rho_{\vec{\sigma}}(h)u_i(h)\geq \sum_{h\in
H^T}\rho_{\sigma',\vec{\sigma}_{-i}}(h)u_i(h)-\epsilon.$$ 
\end{definition}

We now recall the notion of \emph{sequential equilibrium} \cite{KW82}.  
A sequential equilibrium is a pair $(\vec{\sigma},\mu)$ consisting of a
strategy profile $\vec{\sigma}$ and a \emph{belief system} $\mu$, where
$\mu$ associates with each information set $I$ a probability $\mu(I)$ on
the nodes in $I$.  Intuitively, if $I$ is an information set for player $i$,
$\mu(I)$ describes $i$'s beliefs about the likelihood of being in each
of the nodes in $I$.
Then $(\vec{\sigma}, \mu)$ is a sequential equilibrium if, for each
player $i$ and each information set $I$ for player $i$, $\sigma_i$ is a
best response to $\vec{\sigma}_{-i}$ given $i$'s beliefs $\mu(I)$. 
%
An equivalent definition 
that does not require beliefs and is more suitable
for our setting is given by the following theorem: 

\begin{theorem}
\cite[Proposition 6]{KW82} 
Let $G$ be an extensive-form game with perfect recall. 
There exists a belief system $\mu$ such that 
$(\vec{\sigma},\mu)$
is a sequential equilibrium of $G$ iff there exists a sequence of
completely mixed strategy profiles
$\vec{\sigma}^1,\vec{\sigma}^2,\ldots$ converging to $\vec{\sigma}$
and a sequence $\delta_1,\delta_2,\ldots$ of nonnegative real numbers
converging to $0$ such that, for each player $i$ and each information
set $I$ for player $i$, $\vec{\sigma}^n_i$ is a $\delta_n$-best
response to $\vec{\sigma}^n_{-i}$ conditional on having reached $I$. 
\end{theorem}

\newcommand{\eps}{{\epsilon}}
\newcommand{\Gen}{{\text{Gen}}}
\newcommand{\Enc}{{\text{Enc}}}
\newcommand{\Dec}{{\text{Dec}}}
\newcommand{\PPT}{{\text{PPT}}}
\newcommand{\N}{{\mathbb{N}}}

\subsection{Computational indistinguishability}

For a probabilistic algorithm $A$ and an infinite bitstring $r$,
$A(x; r)$ denotes the output of $A$ running on input $x$ with randomness
$r$; 
$A(x)$ denotes the distribution on outputs of $A$ induced by considering
$A(x;r)$, where $r$ is chosen uniformly at random. 
A function $\epsilon :  \mathbb{N} \rightarrow [0,1]$ is\emph{ negligible}
if, for every constant $c \in \mathbb{N}$,  $\epsilon(k) < k^{-c}$ for sufficiently large $k$. 
\begin{definition} A \emph{probability ensemble} is a
sequence $X=\{X_n\}_{n\in\mathbb{N}}$  of probability distribution indexed by
$\mathbb{N}$. 
(Typically, in an ensemble $X=\{X_n\}_{n\in\mathbb{N}}$, the support of $X_n$
consists of strings of length $n$.) 
\end{definition}

We now recall the definition of computational indistinguishability
\cite{goldwasser1984probabilistic}.

\begin{definition} Two probability
ensembles 
$\{X_n\}_{n\in\mathbb{N}}, \{Y_n\}_{n\in\mathbb{N}}$ are \emph{computationally indistinguishable} if, for
all 
PPT TMs $D$, there
exists a negligible function $\epsilon$ such that, for all $n \in
\mathbb{N}$,  
$$|\Pr[D(1^n, X_n) = 1] - \Pr[D(1^n, Y_n) = 1]| \leq \epsilon(n).$$
To explain the $\Pr$ in the last line, recall that $X_n$ and $Y_n$
are probability distributions. Although we write  
$D(1^n, X_n)$, $D$ is a randomized
algorithm, so what $D(1^n, X_n)$ returns depends on the outcome
of random coin tosses.  To be a little more formal, we should write
$D(1^n, X_n,r)$, where $r$ is an infinitely long random bit strong (of
which $D$ will only use a finite initial prefix).  More
formally, taking $\Pr_{X_n}$ to be the joint distribution over strings
$(x,r)$, where $x$ is chosen according to $X_n$ and r is chosen
according to the uniform distribution on bit-strings,  
we want  
$$|{\Pr}_{X_n}\left[\{(x,r):D(1^n, x,r) =
1\}\right]-{\Pr}_{Y_n}\left[\{(y,r):D(1^n, 
y,r) = 1\}\right]|\leq \epsilon(n).$$   
We similarly abuse notation elsewhere in writing $\Pr$.
\end{definition}

We often call a TM $M$ that is
supposed to distinguish between two probability ensembles a
\emph{distinguisher}. 

\subsection{Commitment schemes}
We now define a cryptographic commitment scheme that will be used in
our examples. Informally, such a scheme is a two-phase two-party
protocol for a sender and a receiver. In the first phase, the sender sends a
message to the receiver that commits the sender to a bit
without giving the receiver any information about that bit; 
in the second phase, the sender reveals the bit to which he committed
in a way that guarantees that this really is the bit he committed to.

\begin{definition}
A \emph{secure commitment scheme with perfect bindings} is a pair of PPT
algorithms $C$ and $R$ such that: 
\begin{itemize}
\item $C$ takes as input a security parameter $1^k$, a bit $b$, and a
  bitstring $r$, and 
outputs $C(1^k,b,r),C_2(1^k,b,r)$, where $C_1(1^k,b,r)$, called the
\emph{commitment string}, is a $k$-bit string, and $C_2(1^k,b,r)$, called
the \emph{commitment key}, is a $(k-1)$-bit string. We use $C(1^k,b)$ to denote the output distribution of algorithm $C(1^k,b,r)$ when $r$ is chosen uniformly at random. 
\item $R$ is a deterministic algorithm that gets as input two strings
  $c$ and $s$ and outputs $o\in\{0,1,f\}$. 
\item(Hiding) $\{C_1(1^k,0)\}_{k\in\mathbb{N}}$ and
  $\{C_1(1^k,1)\}_{k\in\mathbb{N}}$ are  computationally
  indistinguishable. 
\item (Perfect binding) $R(C_1(1^k,b,r),(C_2(1^k,b,r))=b$ for all $k$ and $r$;
  moreover, if $s \ne C_2(1^k,b,r)$, then $R(C_1(1^k,b,r)),s) \notin \{0,1\}$.
\end{itemize}
\end{definition}
\noindent
Cryptographers typically assume that secure commitment schemes with
perfect bindings exist.  (Their
existence would follow from the existence of \emph{one-way permutations};  see
\cite{goldreich01} for further discussion and formal definitions.)

\section{Computational Extensive-Form Games}
\subsection{Definitions}\label{sec:definitions}
Statements of computational difficulty typically say that there is no
(possibly randomized) polynomial-time algorithm for solving a problem.
To make sense of this, we need to consider, not just one input, but a
sequence of inputs, getting progressively larger.  Similarly, to make
sense of computational games, we cannot consider a single game, 
but rather must consider a sequence of games that grow in size. 
The games in the sequence share the same basic structure.
This means that, among other things, they involve the same set of
players, playing in the same order, with corresponding utility functions.
To make this precise, we first start with a more general notion, which
we call a \emph{computable uniform sequence of games}.

\begin{definition}
A \emph{computable uniform sequence $\mathcal{G}=\{G_1,G_{2},\ldots\}$
of games} 
is a sequence that satisfies the following conditions:
\begin{itemize}
\item All the games in the sequence involve the same set of players.
\item Let $H_n$ be the set of histories in $G_n$. There exists a
  polynomial $p$ such that, for 
all nonterminal histories $h \in H_n^{NT}$, \hbox{$A(h) \subseteq
  \{0,1\}^{\leq p(n)}$}.\footnote{$\{0,1\}^{\le p(n)}$ denotes the
  language consisting of bitstrings of length at most $p(n)$.} 
In addition, there is a PPT algorithm that, on input $1^n$ and a
history $h$, determines whether $h\in H_n$. 
\item There exists a polynomial-time computable function $P'$ from
$\bigcup_{n=1}^{\infty}(H_n^{NT})$ to $[c]$. The function $P_n$ in game
$G_n\in\mathcal{G}$ is then $P'$ restricted to $H_n^{NT}$. 
\item For each player $i$, there exists a  polynomial-time computable
function $u_i:\bigcup_{n=1}^{\infty}H_n^T \to\mathbb{R}$
such that the utility function of 
player $i$ in game $G_n$ is $u_i$ restricted to $H_n^T$.
\end{itemize}
We sometimes call a computable uniform sequence of games a 
\emph{computational game}.
\end{definition}

Computable uniform sequences of games already suffice to allow us to
talk about polynomial-time strategies.  
%
A strategy $M$ for player $i$
in a computable uniform sequence $\G = (G_1, G_2, \ldots)$
is
%
a probabilistic TM
that 
takes as input a pair $(1^n, v)$, where $v$ is a view for player $i$
in $G_n$ (discussed below),
and outputs
an action in $A(I)$.
\ifstateless
\else
We assume that the TMs are \emph{stateful}; they have a tape on which the random bits used in previous rounds are recorded.
\fi
The \emph{view} of a stateful TM $M$ 
for player $i$ in $G_n$ 
is a tuple $(v_I, r)$, where $v_I$ is the
representation of information set $I$ 
and $r$ contains the randomness that has been used thus far (so is nondecreasing
from round to round).  
This can be viewed as having perfect recall of randomness, as the TMs
are not allowed to ``forget" the randomness they used. It is
considered part of their experience so far in the same way as the
actions that they played and the information sets that they visited.%
\footnote{
This assumption is equivalent to allowing the TM to have an
additional tape on which it can save an arbitrary state. For any TM
$M$ that does this, there is an equivalent TM $M'$ that has no
additional tape, but simply reconstructs $M$'s state by 
simulating $M$'s computation from scratch using its view. 
This suffices,
for example, to reconstruct a secret key that was 
generated in the first 
round, so  it can be used in later rounds. 
}

We next define what it means for a uniform sequence $\G = (G_1, G_2,
\ldots)$ of games to
\emph{represent} an underlying game $G$.  
To explain different aspects of this definition, it is useful to go
back to the example in the introduction and 
discuss what it means for a sequence $\G$ to represent the game $G$ in
Figure~\ref{fig:extensiveGame}.    
As discussed before, we can implement this game
using a commitment scheme.  The point is that now we get, not one
game, but a sequence of games, one for each choice of security
parameter.  Rather that putting a bit $b$ in an envelope, in $G_n$ player 1 sends 
$C_1(1^n,b)$.  More precisely, he sends $C_1(1^n,b,r)$, for a string $r$
chosen chosen uniformly at random.  To then open the envelope, player
$1$ can just send $C_2(1^n,b,r)$ and any other string to destroy it. 

Roughly speaking, we want all the
games in $\G$ to have the same ``structure'' as $G$.  We formalize
this by requiring a surjective mapping $f_n$ from histories in each game 
$G_n$ in the sequence to histories in $G$.  
Note that $f_n$ is not, in general, one-to-one.  There may be
many histories in $G_n$ 
representing a single history in $G$.  This can already be seen in our
example; each of the histories in $G_n$ where player 1 sends
$C_1(1^n,1,r)$ get mapped to the history in $G$ where player 1 puts 1
in an envelope.   Moreover, although $C_1(1^n,0,r)$ and 
$C_1(1^n,1,r)$ get mapped to histories in the same information set
in $G$, they are \emph{not} in the same information set in $G_n$; an
exponential-time player can break the encryption and tell that they
correspond to different bits being put in the envelope.  Thus, the
mapping $f_n$ does not completely preserve the information structure.  
We require that $h$ and $f_n(h)$ have the same length (same number of actions).
Of course, the utility associated with a terminal
history $h$ in $G_n$ is the same as that associated with history
$f_n(h)$ in $G$.  

The first three conditions below capture the relatively
straightforward structural requirements above.  
The final requirement imposes conditions on the players' 
strategies, and is somewhat more complicated.
Informally, the fourth requirement is that there is a mapping $\F$ 
from strategies in $G$ to strategies in $\G$, where $\F(\sigma)$
``corresponds'' to $\sigma$ in some appropriate sense.  
But what should  ``correspond'' mean?
Let $\vec{M}$ be a strategy profile for $\mathcal{G}$. 
For each game $G_n \in \mathcal{G}$, $\vec{M}$ induces a distribution
denoted $\psi^{G_n}_{\vec{M}}$ on the terminal histories in $G_n$.
By applying $f_n$, we can push this forward to a distribution
$\phi^{G_n}_{\vec{M}}$ on the terminal histories in $G$.  
A mixed strategy profile $\vec{\sigma}$ in $G$ also induces a
distribution on the terminal histories in $G$, denoted
$\rho_{\vec{\sigma}}$.    

\begin{definition}\label{def:coresponds}
A strategy profile $\vec{\sigma}$
\emph{corresponds} to $\vec{M}$ if 
$\{\phi^{G_n}_{\vec{M}}\}_{n\in\mathbb{N}}$ is \emph{statistically close} to
$\{ \rho_{\vec{\sigma}}\}_{n\in\mathbb{N}}$: that is, if $H^T$ are the
terminal histories of $G$, then 
there exists a negligible function $\epsilon$ such that, for all $n$,
$$\sum_{h\in H^T}|Pr_{
  \phi^{G_n}_{\vec{M}}}[h]-Pr_{
  \rho_{\vec{\sigma}}}[h]|\leq \epsilon(n).$$ 
\end{definition}

So one requirement we will have is that, for all strategy profiles
$\vec{\sigma}$ in $G$, $\vec{\sigma}$ corresponds to $(\F(\sigma_1),
\ldots, \F(\sigma_n))$, which we abbreviate as $\F(\vec{\sigma})$.
In addition, we require that the strategy profile
$\F(\vec{\sigma})$  ``knows" which
underlying action it plays. We formalize this by requiring that, for
strategy $\sigma$ in the underlying game, there is a TM 
$M^{\sigma}$ that, given 
view $v$ for player $i$ in $\G$, outputs the 
action in $G$ corresponding to the action played by $\F(\sigma)$ given
view $v$. 

Finally, we require a partial converse to the correspondence 
requirement. It is clearly too much to expect a full converse. 
$\G$ has a richer structure than $G$; it allows for more
ways for the players to coordinate than $G$.
So we cannot expect every strategy profile in $\G$ to correspond to a
strategy profile in $G$.
Thus, we require only that strategies in a rather restricted class
of strategy profiles  
in $\G$ correspond to a strategy in $G$: namely,
ones where we start with a strategy of the form $\F(\vec{\sigma})$
(which, by assumption, corresponds to $\vec{\sigma}$),
and allow one player to deviate.  
We must also use a weaker notion of correspondence here.  
For example, in the game in Figure~\ref{fig:extensiveGame}, 
even if we start with a strategy of the form $\F(\vec{\sigma})$, the
deviating strategy 
$M_1'$ 
could be such that player 1 commits to 0 in $G_n$ for $n$ even, and
commits to 1 in $G_n$ for $n$ odd. 
The strategy profile $(M_1',\F(\sigma_2))$ does not correspond to any strategy
profile in $G$.  
Thus, the notion of correspondence that we consider in this case is
that if $i$ plays $M_i'$ rather that $\F(\sigma_i)$, then there exists 
a sequence $\sigma_1', \sigma_2', \ldots$ of
strategies in $G$, rather that a single strategy $\sigma'$,
and require only that the sequence
$\{\phi^{G_n}_{(M_i',\F(\vec{\sigma}_{-i}))}\}_{n\in\mathbb{N}}$ be computationally
indistinguishable from 
$\{ \rho_{\vec{\sigma}}\}_{n\in\mathbb{N}}$, rather than being
statistically close.


\begin{definition}
\label{def:represents}
A computable uniform sequence $\mathcal{G}=\{G_1,G_{2},\ldots\}$
\emph{represents} an underlying game $G$ if the following conditions hold: 
\begin{itemize}
\item[UG1.] $G$ and every game in $\G$ involve the same set of players.
\item[UG2.] For each game $G_n \in \mathcal{G}$, 
there exists a surjective mapping $f_n$ from the histories in $G_n$ to
the histories in $G$ such that
\begin{itemize}
\item[(a)] $|h| = |f_n(h)|$;
\item[(b)] the same player moves in $h$ and $f_n(h)$;
\item[(c)] if $h'$ is a subhistory of
$h$, then $f_n(h')$ is a subhistory of $f_n(h)$;
\item[(d)] 
if $h$ and $h'$ are in the same information set in $G_n$, then $f_n(h)$
and $f_n(h')$ are in the same information set in $G$; 
\item[(e)] for $h\in H$ (a history of $G$), let $\LA(h)$ denote the last action
played in $h$;
 if $h$ and $h'$ are in the same
information set in $G_n$, then for any $a$ such that $h||a\in H_n$, 
$\LA(f_n(h||a))=\LA(f_n(h'||a))$ (where $||$ is the concatenation operator).%
\end{itemize}
\item[UG3.] If $h$ is a terminal history of $G_n$, then the utility of each
player $i$ is the same in $h$ and $f_n(h)$.
\item[UG4.] 
There is a mapping $\F$ from strategies in $G$ to
strategies in $\G$ such that 
\begin{itemize}
\item[(a)] 
for all strategy profiles $\vec{\sigma}$ in $G$, 
$\vec{\sigma}$ corresponds to $\F(\vec{\sigma}) = 
(\F(\sigma_1), \ldots, \F(\sigma_n))$;

\item[(b)] for each strategy $\sigma$ for player $i$ in $G$, 
there exists a 
polynomial-time TM $M^{\sigma}$ that, given as input $1^n$ and a 
view $v$ for player $i$ in $G_n$ that is reachable when player $i$
plays $\F(\sigma_i)$ in $G_n$, returns an action for player $i$ such
that  
$LA(f_n(\F(\sigma)(1^n,v,r_T))) = M^{\sigma}(1^n,v,r_T)$, where $r_T$
is the random tape used (remember that the view contains the
randomness used so far);
\item[(c)]
for all strategy profiles $\vec{\sigma}$ in $G$, all players $i$, and all
polynomial-time strategies $M'_i$ for player $i$ in $\G$, there exists
a sequence $\sigma'_1,\sigma'_2,\ldots$ of strategies for player $i$
in $G$ such that  
$\{\phi^{G_n}_{(M'_i,\F({\vec{\sigma}_{-i}}))}\}_n$ is
computationally indistinguishable from
$\{\rho^{G}_{(\sigma'_n,\vec{\sigma}_{-i})}\}_n$.  
\end{itemize}
\end{itemize}
\end{definition}


Definition~\ref{def:represents} requires the existence of a sequence
$\vec{f}=(f_1,f_2\ldots)$ in UG2 and a function $\F$ in UG4. When we want to
refer specifically to $f$ and $\F$, we say that $\G$
$\langle\vec{f},\F\rangle$-represents $G$. 

Note that UG2 requires that if $h$ and $h'$ are in the same
information set in $G_n$, then $f_n(h)$ and $f_n(h')$ must be in the same
information set in $G$. 
This means that we can view $f_n$ as a map from information sets to
information sets. 
However,
it does \emph{not} require the converse.   As
discussed above, in $\G$, an exponential-time player
may be able to make distinctions between histories that cannot be made
of the corresponding histories in the underlying game.  We would like
to be able to say that a polynomial-time player cannot distinguish $h$
and $h'$ if $f_n(h)$ and $f_n(h')$ are in the same information set.
As we show later, these conditions allow us to make such a claim.

Also note that since the game is finite, to show UG4(a) and UG4(b) hold,
it is enough to prove they hold for deterministic strategies.
Given a mapping $\F$ that satisfies UG4(a) and (b) for deterministic
strategies, we can extend it to mixed strategies in the obvious way:
since a mixed strategy is just a probability distribution over 
finitely many deterministic strategies, it can be implemented 
by a TM that plays that probability distribution up to negligible
precision over the corresponding mapping of the deterministic strategies
(such an approximating distribution can be easily constructed in
polynomial time). It is obvious that UG4(a) still holds. UG4(b) holds
since using $v$ and $r_T$, we can reconstruct which deterministic
strategy $\sigma'$ in the support of $\sigma$ was actually used to reach
$v$, and then use the corresponding TM $M^{\sigma'}$.

\subsection{The commitment game as a uniform computable sequence}\label{sec:commitGame} 
We now consider how these definitions play out in the
game
$G$ in Figure~\ref{fig:extensiveGame} and the sequence
$\mathcal{G} = (G_1, G_2, \ldots)$ described above where player 1 uses
a commitment scheme as an envelope. 

\begin{lemma}\label{lem:coinFlipRep}
$\mathcal{G}$ represents $G$.
\end{lemma}
\begin{proof}
First, we show that $\mathcal{G}$ is a computable uniform
sequence. 
All the games in the sequence involve exactly 2 players; the set of histories
in $G_n$ is a subset of $\{0,1\}^{\leq n}$, and it is easy to compute
the next player to act; finally, the utility functions are
polynomial-time computable by using the TM $R$ of the commitment
scheme.  

Next we show that the sequence represents
$G$. There is an obvious mapping from histories of the games in the
sequence to histories of $G$: a commitment to $0$ is mapped to $0$, a
commitment to $1$ is mapped to $1$, the action of player 2 is just
mapped to the action in $G$, player 1 providing the right key is mapped to
action ``open", and player 1 providing a wrong key is mapped to
``destroy".  Finally, it 
is easy to verify that 
UG3 (the condition on utilities) holds.

To show that UG4 holds, we need to define a function $\F$.
A strategy for player $2$ in $G$ can't depend on player 1's action,
since player 2's information set contains both actions. Thus, a
deterministic strategy $\sigma_2$ for player $2$ in $G$ 
just plays an action in $\{0,1\}$;   
the corresponding strategy
$\F(\sigma_2)$ just plays the same string. 
UG4(b) holds trivially in this case.
To define $\F(\sigma_1)$ for a strategy $\sigma_1$ for player $1$, we
need to show how to implement each action of player 1. To play 
$b$ at the empty history in $G_n$, 1 plays the commitment string
$C_1(1^n,b,r)$, where $r$ is the randomness used by player 1 in the
computation (which is then saved as the TM's state).  
To play the action ``open'', it  computes $k= C_2(1^n,b,r)$; to
play ``destroy'',  
it plays $k\oplus 1$ (a   
string other than the right key).
It is easy to see that UG4(b) holds for strategies of player 1.
Moreover, it is easy to see that $\F(\vec{\sigma})$ corresponds to
$\vec{\sigma}$, so UG4(a) holds.
We extend $\F$ to mixed strategies as described above.

To see that UG4(c) holds, observe that 
a strategy for player 1 in $G_n$ can clearly be mapped to a
strategy in $G$:
At the empty history player 1  has some distribution over
commitments to 0 and commitments to 1.  This clearly maps to 
a distribution over putting 0 and 1 in the envelope.
At the other nodes where player 1 moves, $G_n$ induces a distribution
over correctly revealing the commitment or doing some other action;
again, this clearly maps to a distribution over ``open'' and
``destroy'' in the obvious way.  Since a strategy $M_1'$ for player 1
in $\G$ induces, for all $n$, a strategy $M'_{1,n}$ for player 1 in
$G_n$, we can associate  a sequence $(\sigma'_1, \sigma'_2, \ldots)$
with $M_1'$.  It is easy to check that, for all strategies $\sigma_2$
for player 2 in $G$, 
$\{\phi^{G_n}_{(M'_i,\F(\sigma_2))}\}_n$ is computationally
indistinguishable from $\{\rho^{G}_{(\sigma'_n,\sigma_2)}\}_n$.

We similarly want to associate with each strategy for player 2 in $\G$
a sequence of strategies in $G$.  This is a little more delicate, since 
the information structure in $G_n$ is not the same as that in $G$.
Given a strategy $\sigma_1$ for player 1 in $G$, 
and an arbitrary polynomial-time strategy $M_2$ for player $2$
in $\G$,  let
$P_n(b)$ be the probability that $M_2$ plays $b$ when
$(\F(\sigma_1),M_2)$ is played in $G_n$. Let $\sigma'_n$ be the
strategy 
in $G$
that plays according to the same distribution.  
We now claim that $\{\phi^{G_n}_{(\F({\sigma_1}), M_2)}\}_n$ is
indistinguishable from $\{\rho^{G}_{\sigma_1,\sigma'_n}\}_n$. Assume,
by way of contradiction, that it is not. This can happen only if, for
infinitely many $n$, $M_2$ plays $0$ and $1$ with non-negligibly
different probabilities, depending on whether it is faced with a
commitment to $0$ or a commitment to $1$.
But that means
that, for infinitely many $n$, it can distinguish those two events with
non-negligible probability. This contradicts the assumption that the
commitment scheme is secure. 
\end{proof}

\subsection{Consistent partition structures}\label{sec:compInfSet}

In this section, we discuss the connection between computational
indistinguishability and information structure in games. 
As we saw,  
when going from the game $G$ in Figure~\ref{fig:extensiveGame} to the
game $\G$ that represents it,
we replaced the information set in
$G$ (the use of an envelope) with computational indistinguishability (a 
commitment scheme).  Although the games in $\G$ are perfect
information games, so that the players have complete information about a
 history, 
if player $1$ uses the commitment scheme appropriately,
then player $2$ does not really understand the
``meaning'' of a history (i.e., whether it represents a commitment to 0
or a commitment to 1). 
On the other hand, if player $1$ ``cheats'' by 
using, for example, some low-entropy random string for the commitment,
player $2$ might 
have a strategy that is able to understand the ``meaning'' of its
action. 
Thus, there is a sense in which the information structure of a
computational game depends on the
strategies of the players. 
This dependence on strategies does not exist in standard
games. If each of two histories $h$ and $h'$ in some information set
$I$ for player $i$ has a positive 
probability of being reached by a particular strategy profile, then 
when player $i$ is in $I$, he will not know which of $h$ or $h'$ was
played, even if he knows exactly what strategies are being
played. 
The situation is different for computational games, in a way we now
make precise.

Suppose that $\G = (G_1, G_2, \ldots)$ $\langle\vec{f},\F\rangle$-represents $G$ and
$h$ is a history of $G$,
so that $f_{n}^{-1}(h)$ is the set of
histories of $G_n$ that are mapped to $h$ by $f_n$. 
For a set $H$ of histories of a game $G_n\in\mathcal{G}$, let
$\mathcal{V}_n(H)$ be the set of views that a player can have at
histories in $H$ when $G_n$ is played. 
For a strategy profile
$\vec{M}$ in $\mathcal{G}$, let $\xi^{G_n}_{\vec{M}}(v)$
be the probability of reaching view $v \in \mathcal{V}_n(H)$ if the players
play strategy profile $\vec{M}$ in $G_n$. For a set $V$ of views, let
$\xi^{G_n}_{\vec{M}}(V)=\sum_{v\in V}\xi^{G_n}_{\vec{M}}(v)$. 
For a set $V$ of mutually incompatible views (i.e., a set $V$ of views
such that for all distinct views $v,v' \in V$, the probability of reaching
$v$ given that $v'$ 
is reached is $0$, and vice versa), 
let $X_{\vec{M},n}^V$ be a probability distribution
on $V$ such that $X_{\vec{M},n}^V(v) =
\frac{\xi^{G_n}_{\vec{M}}(v)}{\xi^{G_n}_{\vec{M}}(V)}$  
if $\xi^{G_n}_{\vec{M}}(V)>0$, and $\frac{1}{|V|}$ otherwise.
Let $\xi_{\vec{\sigma}}^G(S)$ denote the probability of reaching a set $S$
of histories in $G$
if the players play strategy profile $\vec{\sigma}$.
Note that if
$\xi_{\vec{\sigma}}^G(S)>0$, then by UG4, for 
all
sufficiently large $n$,
we must have 
$\xi^{G_n}_{\vec{M}_{\vec{\sigma}}}(\mathcal{V}_n(f_{n}^{-1}(S)))>0$.

\begin{definition}
Let $\G$ $\langle\vec{f},\F\rangle$-represent a game $G$
and let $\vec{M}$ be a strategy in $\G$.  
A partition $\I_i$ of $H^i$ (recall that this is the set of histories
in $G$
where $i$ plays) is 
\emph{$\vec{M}$-consistent} for player $i$ if, for all
non-singleton $I\in \I_i$ and all $h\in I$ such that both
$\xi^{G_n}_{\vec{M}}(\mathcal{V}_n(f_{n}^{-1}(h)))$ and
$\xi^{G_n}_{\vec{M}}(\mathcal{V}_n(f_{n}^{-1}(I\setminus
  h)))$ are non-negligible,
$\{X_{\vec{M},n}^{\mathcal{V}_n(f_{n}^{-1}(h))}\}_{n\in\mathbb{N}}$ 
is computationally  
indistinguishable from $\{X_{\vec{M},n}^{\mathcal{V}_n(f_{n}^{-1}(I\setminus \{h\}))}\}_{n\in\mathbb{N}}$.
A partition structure $\vec{I}$ is $\vec{M}$-consistent if, for all
agents $i$,  
$\vec{\I}_i$ is $\vec{M}$-consistent.
\end{definition}

Intuitively, a partition $\I_i$ for player $i$ is consistent with a
strategy profile $\vec{M}$, if, when $\vec{M}$ is played in $\G$, 
for all $I \in \I_i$ and all histories
$h, h' \in I$, the distribution over views that map to $h$ is
computationally indistringuishable from the distribution over views
that map to $h'$.  
In our example, 
this means that player $2$ can't distinguish between the
distribution created by a commitment to $0$ and the distribution created
by a commitment to $1$ if the commitment algorithm is run ``honestly''
(using truly random strings). 

Note that we do not enforce any condition on histories in $G$ that are mapped
back to a set of histories that is reached with only negligible
probability. This means there might be 
more than one $\vec{M}$-consistent information partition. 

We next show that if $\I_i$ is the information partition of player $i$
in $G$, and $\G$ $\langle\vec{f},\F\rangle$-represent $G$ then for any
strategy profile $\vec{\sigma}$ in $G$, $\I_i$ must be 
$\F({\vec{\sigma}})$-consistent.

\begin{theorem}\label{thm:Mconsistent}
If $\mathcal{G}$ $\langle\vec{f},\F\rangle$-represents $G$, $\I_i$ is the
information partition of player $i$ in $G$, and 
$\vec{\sigma}$ is a strategy profile in $G$ then 
$\I_i$ is $\F(\vec{\sigma})$-consistent. 
\end{theorem}

\begin{proof}
We must show that if
$I\in\I_i$ is a non-singleton information set for $i$ in $G$ and 
$h\in I$ , then for all strategy profiles $\vec{\sigma}$ in $G$
such that 
$\xi_{\vec{\sigma}}^G(h)>0$ and $\xi_{\vec{\sigma}}^G(I\setminus \{h\})>0$,
$\{X_{\F(\vec{\sigma}),n}^{\mathcal{V}_n(f_{n}^{-1}(h))}\}_{n\in\mathbb{N}}$
is computationally 
indistinguishable from
$\{X_{\F(\vec{\sigma}),n}^{\mathcal{V}_n(f_{n}^{-1}(I\setminus
  \{h\}))}\}_{n\in\mathbb{N}}$.   

Assume, by way of contradiction, that $h\in I$, $I$ is an information
set for player $i$ in $G$, and there exists a strategy profile 
$\vec{\sigma}$ in $G$ that
reaches both $h$ and $I\setminus\{h\}$ with positive probability such that
$\{X_{\F(\vec{\sigma}),n}^{\mathcal{V}_n(f_{n}^{-1}(h))}\}_n$ is distinguishable from
$\{X_{\F(\vec{\sigma}),n}^{\mathcal{V}_n(f_{n}^{-1}(I\setminus\{h\}))}\}_n$.
Thus, there exists a distinguisher $D$ for these distributions.
Let $a$ and $a'$ be distinct actions in $A(I)$.  (Recall that we 
assumed that $|A(I)| \ge 2$.)
Let $M'$ be a strategy for player
$i$ in $\G$ such that when $M'$ reaches a history that maps to $I$ 
(by UG4(b) and the fact that the sets of actions available in each
information set are disjoint, this can be checked in polynomial time), 
$M'$ uses $D$ to distinguish if its view
is in $\mathcal{V}_n(f_n^{-1}(h))$ or  
$\mathcal{V}_n(f_n^{-1}(I\setminus\{h\}))$. 
$M'$ then plays an action mapped
to $a$ if $D$ returns $0$ and 
an action mapped to $a'$ otherwise. 
At a history other than one in $f_n^{-1}(I)$, 
$M'$
plays like $\F(\sigma_i)$.  
It is easy to see that, because
$\{X_{\F(\vec{\sigma}),n}^{f_{n}^{-1}(h)}\}_n$ and
$\{X_{\F(\vec{\sigma}),n}^{f_{n}^{-1}(I\setminus\{h\})}\}_n$ are
distinguishable with non-negligible probability, there is
a non-negligible probability that the strategy $M'$ is able to detect 
which case holds, and play accordingly. That means that when 
histories of $(M',\F(\sigma_{-i}))$ are mapped 
to histories of $G$ via $f_n$, there is a non-negligible gap between the
probability of $(h,a)$ and the probability of $(h',a)$ for $h'\in
I\setminus\{h\}$. 
Since $h\in I$, there can
be no strategy $\sigma'$ for player $i$ such that $(\sigma',\sigma_{-i})$
has such a gap, and UG4(c) cannot hold.  This gives us the desired contradiction.
\end{proof}

Note that Theorem~\ref{thm:Mconsistent} holds trivially if, for all
$G_i\in\G$, 
all the histories of $\G$ that map to $I$ are in the same information
set in $G_i$.  
The theorem is of interest only when this is not the case.
If we think of $G$ as an abstract model of a computational game $\G$ that
represents it, this result  
can be thought of as saying that information sets in $G$ can model
both real lack of information and
computational indistinguishability in $\G$.  


\section{Solution Concepts for Computational Games}\label{sec:NE}

In this section, we consider analogues of two standard solution
concepts in the context of computational games: Nash equilibrium and
sequential equilibrium, and prove that they exist if the computational
game represents a finite extensive-form game. 


\subsection{Computational Nash equilibrium}\label{sec:compNash}
Informally, a strategy profile in $\G$ is a computational Nash
equilibrium if no player $i$ has a profitable \emph{polynomial-time}
deviation, where a deviation is taken to be profitable if it is
profitable in infinitely many games in the sequence.
Recall that $\psi^{G_n}_{\vec{M}}$ is the distribution on the terminal
histories in $G_n$ 
induced by a strategy profile $\vec{M}$ in $\G$.

\begin{definition}
$\vec{M}=\{M_1,\ldots,M_c\}$  is a \emph{computational Nash equilibrium}
  of a computable 
uniform sequence  $\mathcal{G}$ if, 
for all players $i\in[c]$ and for all
polynomial-time strategies $M'_i$ in $\G$ for player $i$, there exists
a negligible function $\epsilon$, such that for all $n$,
$$\sum_{h\in H^T_n}\psi^{G_n}_{\vec{M}}(h)u_i(h)\geq \sum_{h\in
H_n^T}\psi^{G_n}_{(M',\vec{M}_{-i})}(h)u_i(h)-\epsilon(n).$$ 
\end{definition}

Our definition of computational NE is similar in spirit to that of
Dodis, Halevi, and Rabin \citeyear{DHR00}, although they formalize it
by having the strategies depend on a security parameter and the
utilities depend only on actions in a single 
normal-form game (rather
than a sequence of extensive-form games).  
Our definition (and theirs) differs from the standard
definition of $\epsilon$-NE in two ways.  First, we restrict 
to polynomial-time deviations.  This
seems in keeping with our focus on polynomial-time players.
Second, we have a negligible loss of utility $\epsilon$ in the
definition, and $\epsilon$ depends on the deviation.  
(The fact that $\epsilon$ depends on the deviation means that what we
are considering cannot be considered an $\epsilon$-Nash equilibrium in
the standard sense.)
Of course, we could have given a definition more in the spirit of the
standard definition of Nash equilibrium by simply taking $\epsilon$ to
be 0.  However, the resulting solution concept would simply not be
very interesting, given our restriction to polynomial-time players.
In general,
there will not be a ``best" polynomial-time strategy;
for every polynomial-time TM, there may be another TM that is better and
runs only slightly longer.  For example, player 2 may be able to do a
little better by spending a little more time trying to decrypt the
commitment in a commitment scheme. 
(See also the examples in \cite{HPR15}.)%
\commentout{
\footnote{One way to avoid having $\epsilon$ depend on the deviation, 
which we do not explore in this paper, is
to instead use a concrete model of complexity in which we use only TMs
with running time less than some specific function $T$ of $n$. In that
case,  we could 
use a single global $\epsilon$, and get a definition that is closer
to that of traditional $\epsilon$-NE.}
}

We now show that our model allows us to provide
conditions that guarantee the existence of 
a computational NE;  to the best of our knowledge,
this has not been done before (and is mentioned as an open question
in~\cite{Katz08}).  
More specifically, we show that if a computational game $\G$ represents $G$, then for
every NE $\vec{\sigma}$ in $G$, there is a corresponding 
NE in $\G$.

\begin{theorem}\label{eqExists}
If $\G$ $\langle\vec{f},\F\rangle$-represents $G$ and  $\vec{\sigma}$ is a NE in $G$, then
$\F(\vec{\sigma})$ is a computational NE of $\mathcal{G}$.  
\end{theorem} 
\begin{proof}
Suppose that $\vec{\sigma}$ is a NE in $G$.  
By UG4, $\vec{\sigma}$ corresponds to $\F(\vec{\sigma})$.
Thus, there exists some negligible function $\epsilon$ such that, for all $n$,
$$\sum_{h\in
  H^T}\phi^{G_n}_{\F(\vec{\sigma})}(h)u_i(h)>\sum_{h\in
  H^T}\rho^{G}_{\vec{\sigma}}(h)u_i(h)-\epsilon(n).$$
We claim that $\vec{M}_{\vec{\sigma}}$ is a computational NE of $\G$.
Assume, by way of contradiction, that it is not.
That means there is some player $i$,
some strategy $M'_i$ for player $i$, and some constant $c > 0$ such that,
for infinitely many values of $n$,
$$\sum_{h\in
  H^T}\phi^{G_n}_{(M',\F(\vec{\sigma}_{-i}))}(h)u_i(h)>\sum_{h\in
 H^T}\phi^{G_n}_{\F(\vec{\sigma})}(h)u_i(h)+\frac{1}{n^c};$$ 
If not, we could have constructed a negligible function to satisfy
the equilibrium condition. 

By combining the two equation we get that for infinitely many values of $n$,
$$\sum_{h\in
  H^T}\phi^{G_n}_{(M',\F(\vec{\sigma}_{-i}))}(h)u_i(h)>\sum_{h\in 
  H^T}\rho^{G}_{\vec{\sigma}}(h)u_i(h)-\epsilon(n)+\frac{1}{n^c}.$$
%
Since $\vec{\sigma}$ is a NE, we get that for
all sequences  $\sigma_1',\sigma_2'\ldots$ of strategies for player
$i$ in $G$, 
$$\sum_{h\in
  H^T}\rho^{G}_{\vec{\sigma}}(h)u_i(h)\geq\sum_{h\in
  H^T}\rho^{G}_{(\sigma'_n,\vec{\sigma}_{-i})}(h)u_i(h).$$
This means that for infinitely many values of $n$,  and for any such sequence,
$$\sum_{h\in
H^T}\phi^{G_n}_{(M',\F({\sigma}_{-i}))}(h)u_i(h) > \sum_{h\in
  H^T}\rho^{G}_{(\sigma'_n,\vec{\sigma}_{-i})}(h)u_i(h)-\epsilon(n)+\frac{1}{n^c}.$$
%
But this contradicts UG4(c), which says that there must exist a sequence
$\sigma_1',\sigma_2'\ldots$ such that
$\{\phi^{G_n}_{(M'_i,\F(\vec{\sigma}_{-i}))}\}_n$ is 
computationally indistinguishable from
$\{\rho^{G}_{(\sigma'_n,\vec{\sigma}_{-i})}\}_n$. 
Since the difference between the two payoffs is not negligible, a
distinguisher could just sample enough outcomes of these
strategies and compute the average
payoff to distinguish the two distributions with non-negligible
probability. Thus, $\vec{M}_{\vec{\sigma}}$ must be a computational NE of $\G$. 
\end{proof}



Theorem~\ref{eqExists} shows that every NE in $G$ has a 
corresponding NE in $\G$.  The converse does not hold.
This should not be surprising; the set of strategies in $\G$ is much
richer than that in $G$.  The following example gives a simple
illustration.

\begin{exm}
{\rm
Consider the  2-player game $G'$ that is like the game in
Figure~\ref{fig:extensiveGame}, except that the payoff is 1 to both if
they match and 0 otherwise 
(and both get $-1$ if player 1 does not open the envelope). 
This game has three NE: both play 0; both play 1; and both play the
mixed strategy that gives probability $1/2$ to each of 0 and 1.
%
There is a computational game $\G'$ that represents $G'$ that is
essentially identical to the game $\G$ described in
Section~\ref{sec:commitGame}, except that the payoffs are modified
appropriately.  
The game $\G'$ has many more equilibria than $G'$, since player 1 can
commit to $0$ 
and $1$ with 0.5 probability but use a fixed key that the second
player knows (or choose a random key from a low entropy set that the
second player can enumerate). Player 2 can take advantage of this to
always play the matching action. There is no strategy in
$G'$  that can mimic this behavior.   
}
\end{exm}

\subsection{Computational sequential equilibrium}\label{sec:compSeq}
Our goal is to define a notion of computational sequential
equilibrium.  To do so, it is useful to think about the standard
definition of sequential equilibrium at an abstract level.   
Essentially, $\vec{\sigma}$ is a sequential equilibrium if, for each player 
$i$, there is a partition $\I_i'$ of the histories where $i$ plays such
that, at each cell $I \in \I_i'$, player $i$ has beliefs about the
likelihood of being at each history in $I$, and the action that he
chooses at a history in $I$ according to $\sigma_i$
is a best response, given these beliefs and what the other agents are
doing (i.e., $\sigma_{-i}$).  The standard definition of sequential
equilibrium takes the partition $\I_i'$ to consist of $i$'s
information sets.  If we partition the histories into singletons, we
get a \emph{subgame-perfect equilibrium}~\cite{Selten65}.  As we
argued in Section~\ref{sec:compInfSet}, the information sets  sets in $\G$ are
too fine, in general, to capture a player's ability to distinguish.
Thus, as a first step to getting a notion of computational sequential
equilibrium, we generalize the standard definition of sequential equilibrium 
in a straightforward way to get $\vec{\I}$-sequential equilibrium,
where $\I_i$ is an arbitrary partition of the histories where $i$ plays.

\begin{definition}
Given a partition $\vec{\I}$, $\vec{\sigma}$ 
is a \emph{$\vec{\I}$-sequential equilibrium} of $G$
if there exists a sequence of 
completely mixed strategy profiles
$\vec{\sigma}^1,\vec{\sigma}^2,\ldots$ converging to $\vec{\sigma}$
and a sequence $\delta_1,\delta_2,\ldots$ of nonnegative real numbers
converging to $0$ such that, for each player $i$ and each
set $I\in\I_i$, $\vec{\sigma}^n_i$ is a $\delta_n$-best
response to $\vec{\sigma}^n_{-i}$ conditional on having reached $I$. 
\end{definition}

What are reasonable partition structures to use when considering a
computational game? As we suggested, using the information partition
structure of $\G$ seems unreasonable. For example, in our example
commitment game, this does not allow the second player to act the same
when facing commitments to $0$ and commitments to $1$, 
although, as we argued earlier, if player 1 plays appropriately, 
a computationally bounded player cannot distinguish these two events.  

It seems reasonable to have histories in 
the same cell of the partition if the player cannot distinguish what
these histories actually ``represent".  For general
uniform computable sequences it is unclear what ``represents" should mean. 
However, if $\G$ represents a game $G$, then we do have in some sense a
representation for a history: the history it maps to in the
underlying game. As we saw in Section~\ref{sec:compInfSet}, what
a player can infer from a history might depend not just on the
information partition structure of the games in $\G$, but also on the
strategies played by the players in $G$. Thus, a natural candidate for
a partition structure $\vec{I}$ when $\vec{M}$ is the strategy profile
played is a partition that is based on an $\vec{M}$-consistent
partition structure $\vec{\I}_G$ of the histories of $G$.
We now formalize this intuition. 

Suppose that $\G$ $\langle\vec{f},\F\rangle$-represents $G$. Given a set
$I\subseteq H$, let $I_{G_n}$ be the set consisting of histories $h \in 
G_n$ such that $f_n(h) \in I$. 
Given two strategies $M$ and $M'$ for a player in $\G$, 
let $(M,I,M')$ be the TM that plays like $M$
in $G_n$ up to  $I_{G_n}$, and then switches to playing $M'$ from that
point on. 
For a game $G_n\in\G$, 
a strategy profile $\vec{M}$, and a set $H'_n$ of histories in $G_n$
that is reached with positive probability when $\vec{M}$ is played,
let $\phi_{\vec{M},H'_n}^{G_n}$ be the probability on terminal histories
in $G$ induced by pushing forward the probability on terminal
histories in $G_n$ conditioned on reaching $H'_n$ (where we identify
the event ``reaching $H'_n$" with the set of terminal histories that
extend a history in $H'_n$).  We can similarly define
$\rho_{\vec{\sigma},H'}^G$ for a subset $H'$ of histories in $G$.

\begin{definition}
Suppose that $\G$ $\langle\vec{f},\F\rangle$-represents $G$.  Then $\vec{M}=\{M_1,\ldots,M_c\}$ 
is a \emph{computational sequential equilibrium} of $\G$ if there
exists a sequence  
of completely mixed strategies $\vec{M}^1,\vec{M}^2,\ldots$ converging
to $\vec{M}$
 and a 
sequence $\delta_1,\delta_2,\ldots$ converging to $0$ such that, for
all $k$, $n$, and players $i\in[c]$, 
there exists an $\vec{M}$-consistent partition $\I_i$ such that, for
all sets $I\in \I_i$ and 
all polynomial-time strategies $M'$ for player $i$ in $\G$, there exists a
negligible function $\epsilon$ such that 
$$\sum_{h\in H^T}\phi^{G_n}_{\vec{M}^k,I_{G_n}}(h)u_i(h)\geq \sum_{h\in
H^T}\phi^{G_n}_{((\vec{M}^k_i,I,M'),\vec{M}^k_{-i}),I_{G_n}}(h)u_i(h)-\epsilon(n)-\delta_k.$$    
\end{definition}

We now claim that, as with NE, if $\vec{\sigma}$ is a sequential
equilibrium of an extensive form game $G$ with perfect recall and $\G$
$\langle\vec{f},\F\rangle$-represents $G$, then 
$\F(\vec{\sigma})$ is a computational sequential equilibrium of $\G$. 

\begin{theorem}\label{seqeqrep}
Suppose that $\G$ $\langle\vec{f},\F\rangle$-represents $G$ and $G$ has perfect recall. 
If there exists a belief function $\mu$ such that 
$(\vec{\sigma},\mu)$ is a sequential equilibrium in $G$, then
$\F(\vec{\sigma})$ is a computational sequential equilibrium of
$\mathcal{G}$.   
\end{theorem}

\begin{proof}
Suppose that there exists a belief system $\mu$ such that
$(\vec{\sigma},\mu)$ is a sequential equilibrium in $G$. Thus, 
there exists a sequence of completely mixed strategy profiles $\vec{\sigma}^1,\vec{\sigma}^2,\ldots$ that
converges to $\vec{\sigma}$ and a sequence $\delta_1,\delta_2,\ldots$
that converges to $0$ such that for all players $i$, 
all information sets $I$ for $i$ in $G$, and all strategies $\sigma'$
for $i$ that act like $\sigma$ on all prefixes of histories in $I$, we
have that    
\begin{equation}\label{eq:seqeq2}
 \sum_{h\in
  H^T}\rho_{\vec{\sigma}^k,I}^G(h)u_i(h)\geq \sum_{h\in
 H^T}\rho_{(\sigma',\vec{\sigma}^k_{-i}),I}^G(h)u_i(h)- \delta_k.
\end{equation}
Assume, by way of contradiction, that
$\vec{M}=\F(\vec{\sigma})$ is not a computational sequential equilibrium. Let
$M_i^k$ be the TM that acts like $\F(\sigma_i^k)$ except that at
a view it is called to play, with probability $\frac{1}{2^{nk}}$ (which is
negligible), it 
plays an arbitrary legal action, chosen uniformly at random.  Note
that this makes $M_i^k$ completely mixed, while ensuring that $\vec{M}^k$
still corresponds to $\vec{\sigma}^k$.
Also note that the sequence $\vec{M}^1,\vec{M}^2,\ldots$ converges to
$\vec{M}$.
By Theorem~\ref{thm:Mconsistent}, if $\I_i$ is the
information partition of player $i$ in $G$, then $\I_i$ is
$\vec{M}^k$-consistent for all $k$,
and, in particular, is also $\vec{M}$-consistent.
Since $\vec{M}$ is not a computational sequential equilibrium, 
there must be some $k$, player $i$, information set $I$ for $i$
in $G$, 
strategy $M'_i$ for  $i$, and constant $c$ such that, for infinitely
many values of $n$, 
\begin{equation}\label{eq:seqeq3}
\sum_{h\in
  H^T}\phi^{G_n}_{((\vec{M}^k_i,I,M'),\vec{M}^k_{-i}),I_{G_n}}(h)u_i(h)>\sum_{h\in
  H^T}\phi^{G_n}_{\vec{M}^k,I_{G_n}}(h)u_i(h)+\frac{1}{n^c}+\delta_k.
\end{equation}

Since $\vec{\sigma}^k$ is completely mixed, every terminal history 
is reached with positive probability.  Thus,
$I_{G_n}$ is reached with positive probability.
Since $\vec{M}^k$ corresponds to $\vec{\sigma}^k$, 
$\{\phi^{G_n}_{\vec{M}^k,I_{G_n}}\}_n$ (the conditional ensemble) must be
statistically close to $\{\rho^{G}_{\vec{\sigma}^k,I}\}_n$, 
for otherwise we could use the distinguisher for these ensembles to
distinguish the unconditional ensembles.
It follows that
there exists some negligible function $\epsilon$ such that, for
all $n$, 
\begin{equation}\label{eq:seqeq4}
\sum_{h\in
  H^T}\phi^{G_n}_{\vec{M}^k,I_{G_n}}(h)u_i(h)>\sum_{h\in
  H^T}\rho^{G}_{\vec{\sigma}^k,I}(h)u_i(h)-\epsilon(n).
\end{equation}
From (\ref{eq:seqeq3}) and (\ref{eq:seqeq4}), it follows that, for infinitely many values of $n$,
\begin{equation}\label{eq:seqeq5}
\sum_{h\in
H^T}\phi^{G_n}_{((\vec{M}^k_i,I,M'),\vec{M}^k_{-i}),I_{G_n}}(h)u_i(h)>\sum_{h\in 
H^T}\rho_{\vec{\sigma}^k,I}^G(h)u_i(h)-\epsilon(n)+\frac{1}{n^c}+\delta_k.
\end{equation}
%
By UG4(c), there is a sequence 
$\sigma'_1,\sigma'_2,\ldots$ of strategies for $i$ in $G$ such that
$\{\phi^{G_n}_{((\vec{M}^k_i,I,M'),\vec{M}^k_{-i}})\}_n$ is
computationally indistinguishable from
$\{\rho^{G}_{(\sigma'_n,\vec{\sigma}^k_{-i})}\}_n$. 
Since, for $n$ sufficiently large, $I_{G_n}$ is reached
with non-negligible probability by $\vec{M}^k$, and
$(\vec{M}^k_i,I,M')$ acts like $\vec{M}^k_i$ in all prefixes of
histories in $I_{G_n}$, it must be the case that for $n$ sufficiently large,
$((\vec{M}^k_i,I,M'),\vec{M}^k_{-i})$ reaches $I_{G_n}$ with non-negligible
probability. Moreover,
$\{\phi^{G_n}_{((\vec{M}^k_i,I,M'),\vec{M}^k_{-i}),I_{G_n}}\}_n$ is
computationally indistinguishable from
$\{\rho^{G}_{(\sigma'_n,\vec{\sigma}^k_{-i}),I}\}_n$. If not, again, a
distinguisher for the unconditional distributions can just use the
distinguisher for the conditional distribution by calling it only when
the sampled history is such that $I$ is visited.  
From (\ref{eq:seqeq2}) and (\ref{eq:seqeq5}), we get that for
infinitely many values of $n$, 
$$\sum_{h\in
  H^T}\phi^{G_n}_{((\vec{M}^k_i,\I(I),M'),\vec{M}^k_{-i}),I_{G_n}}(h)u_i(h)>\sum_{h\in 
  H^T}\rho_{(\sigma'_n,\vec{\sigma}^k_{-i}),I}^G(h)u_i(h)(h)u_i(h)-\epsilon(n)+\frac{1}{n^c}.$$ 
%
But, as in previous arguments, this contradicts the assumption that
$\{\phi^{G_n}_{((\vec{M}^k_i,I,M'_i),\vec{M}^k_{-i}),I_{G_n}}\}_n$  is 
computationally indistinguishable from
$\{\rho^{G}_{(\sigma'_n,\vec{\sigma}^k_{-i}),I}\}_n$. 
Thus, $\vec{M}_{\vec{\sigma}}$ is a computational sequential
equilibrium of $\G$. 
\end{proof}

What are the beliefs represented by this equilibrium? The beliefs we
get are such that the players believe that,
except with negligible probability,
only strategies that are
mappings (via $\F$) of strategies in the underlying game
were used, so they explain deviations in the computational game
in terms of deviations in the 
underlying game. 

One consequence of using completely mixed strategies in the standard
setting is that a player always assigns positive probability to
wherever he may find himself.  
In our setting, while we also require strategies to be
completely mixed, a player $i$ may still find himself in a situation
(i.e., may have a view) to which he ascribes probability 0, so he
knows his beliefs are bound to be incorrect.   
This can happen only if the randomness in $i$'s state is inconsistent
with the moves that $i$ made that led to the current view.  (This can
happen if, for example, $i$ ignored the random string when computing
the commitment string, and just outputted a string of all 1's.) 
%
While $i$ may ascribe probability 0 to his earlier moves, 
deviations by other players always result in views to which $i$ ascribes
positive probability, so such deviations can not be used as signals or
threats.

By considering a consistent partitions here, we effectively
average the expected payoff over all histories of $\G_n$ that map to
the same information set in $I$. Note that, for each specific history
in this set, there might be a better TM. For example, in the commitment
game discussed before, for each commitment string, there is a TM for
player $2$ that does better then the prescribed protocol: the one that
plays the right value given that string. However, our notion considers
the expected value over all these histories, and thus a good deviation
does not exist. Since no 
polynomial-time TM can tell to which histories in the underlying game
these histories are mapped (via $f$), we
treat cells in a consistent partition just as traditional information
sets are treated in the standard notion of
sequential equilibrium. 

\section{Application: Implementing a Correlated Equilibrium Without a
  Mediator}\label{sec:correlated}

In this section, we show that our approach can help us analyze
protocols that use cryptography to implement a correlated equilibrium (CE) 
in a normal-form game.
Dodis, 
Halevi, and Rabin~\citeyear{DHR00} (DHR) were the first to use
cryptographic techniques to implement a CE.
They did so using a protocol that they showed was a NE,
provided that players are computationally bounded (for a
notion of computational NE that is related to ours).  
However, as discussed by Gradwohl, Livne, and Rosen~\citeyear{GLR13} 
(GLR), DHR's proposed protocol does not satisfy solution concepts that also
require some sort of sequential rationality. DHR's protocol punishes deviations
using a minimax strategy that may give the punisher as well as the player
being punished a worse payoff; 
thus, it is just an empty threat. To deal with this issue, GLR
introduce a solution concept that they call \emph{Threat Free Equilibrium
(TFE)}, which explicitly eliminates such empty threats.  GLR 
additionally provide a protocol that can implement a CE in a
normal-form game that is a convex combination of NEs (CCNE), without
using a mediator; the GLR protocol is a TFE if the players  
are computationally bounded. 

We now provide a protocol similar in spirit to the one used in
GLR that implements a CCNE; our protocol is a computational sequential
equilibrium if the 
players are computationally bounded. Unlike GLR, we are able to apply
our approach to CEs in games with more than 2 players,
as well as being able to implement CCNEs that are not Pareto optimal.  
\commentout{
Since GLR's approach applies only to games of perfect information,
they need to use special techniques to deal with the underlying
normal-form game where the CE is played.  Since computational
sequential equilibrium applies to games of imperfect information, we
can view the normal-form game as an extensive-form game where players
move in turn (using information sets to capture the assumption that
the second player does not know the first player's move).
}
One more advantage of our approach is that since we allow the
underlying game to be one of imperfect information, there is a natural
way to model a normal-form game (where players are assumed to move
simultaneously) as an extensive-form game: players just move
sequentially without learning what the other player does.
Since GLR's results apply only to games of perfect information, they
had to argue that they could extend their result to normal-form games.

We require that the CCNE is of finite support, that all its
coefficients are rational numbers, and that each of
the NEs in its support has coefficients that are rational numbers.%
\footnote{GLR also made these assumptions.  In fact, they required a
slightly stronger condition; they required 
  all the coefficients to be rational numbers whose denominator is a power of
  two.} 
We call such CCNEs \emph{nice}. Note that any CCNE can be approximated to
arbitrary accuracy by a nice CCNE. 

\newcommand{\corr}{\mathit{corr}}
Given a normal-form game $G$ with a nice CCNE $\pi$, we show how to convert it
to an extensive-form game 
$G_{\corr}$ that implements this CE without using cryptography, but
using envelopes; that is, $G_{\corr}$ has a sequential
equilibrium with the same distribution over outcomes in $G$ as
$\pi$.  We then show 
how to implement $G_{\corr}$ as a computational game using a
cryptographic protocol.  

Given $G$ and $\pi$, let $\ell$ be the least common denominator of
the coefficients of $\pi$.  Let $G_{\corr}$ be the game where
player $1$ first  
puts an element of $\{0, \ldots ,\ell-1 \}$ in an envelope, then player
$2$ plays an element in $\{0,\ldots ,\ell-1 \}$ without knowing what
player $1$ played (all the histories where player 2 makes his first
move are in the same information set of player 2). Then player $1$ can
either open the envelope or destroy it. All 
the histories after player $1$ opens the envelope form singleton
information sets for the other players; 
all histories after player 1 destroys the envelope
and 2 initially played $j$ are in the same information set for the players  
other than 1, for $j \in \{0,\ldots, \ell-1\}$.  
Then $G$
is played.  (Note that $G$ might involve many players other than 1 and
2, but 1 and 2 are the only players who play in the initial part of
$G_{\corr}$.) 
The players move sequentially: first player 2 moves, then player 1
moves (without knowing player 2's move), then player 3 moves (without
knowing 1 and 2's moves), and so on.
The payoffs
of $G_{\corr}$ depend only on the players' moves when playing the $G$
component of $G_{\corr}$, and are the same as the payoffs in $G$.
See Figure~\ref{fig:corrGame} for a game $G_{corr}$ when $\ell$ is $2$ and
$G$ is a coordination game: that is, in $G$, each player moves
either left or right, and each gets a payoff of 1 if they make the
same move, and -1 if they make different moves.

\begin{figure}
                \centering
                \includegraphics[width=0.6\textwidth]{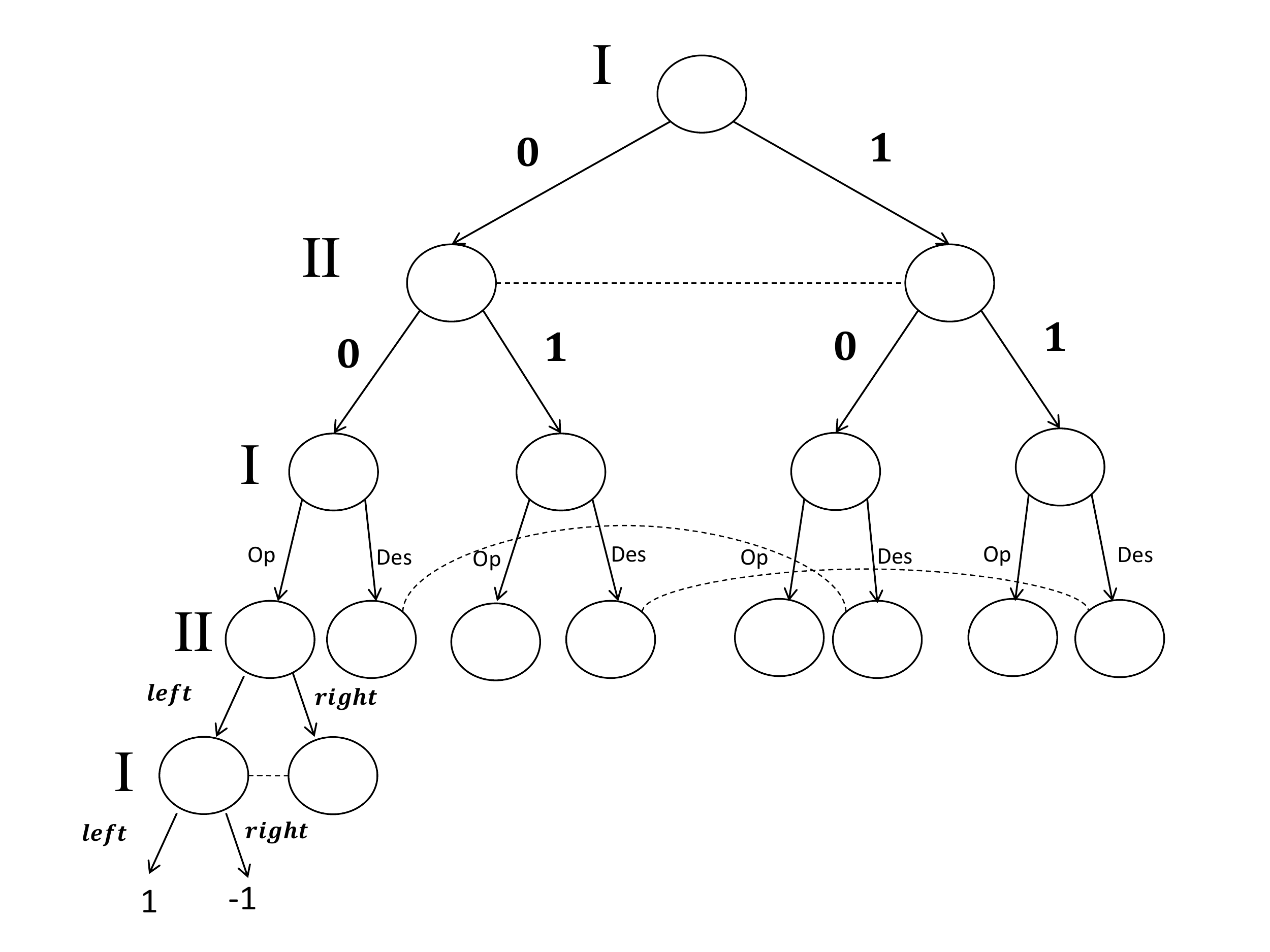}
 	\caption{An example of the game $G_{corr}$ where $\ell=2$ and $G$ is a coordination game}
	\label{fig:corrGame}
\end{figure}

Let $\sigma$ be a NE in $G$ in which player $1$'s payoff is no
better than it is in any other NE
in $G$. 
%
Now consider the following simple strategies for the 
players in $G_{\corr}$.
Intuitively, the players start by picking a NE in the support of
$\pi$ to play,  with probability proportional to its coefficient in
$\pi$.  To this end, 
fix an ordering of length $\ell$ of the NEs in the support of $\pi$,
where each NE appears a number of times proportional to its weight in
the convex combination that makes up $\pi$.
At the empty history, player $1$ selects
an action $a$ 
uniformly at random from \hbox{$\{0, \ldots ,\ell-1 \}$} and puts it
in the envelope. Then player
$2$ also selects an action $b$ uniformly at random from \hbox{$\{0,
  \ldots ,\ell-1 \}$}. Then player $1$ opens the envelope. 
The players
then play the NE in place \hbox{($a+b \mod \ell$)} in 
the ordering
of NEs.
If player $1$ does not open, the players play
according to $\sigma$.  Call the resulting strategy profile $\vec{\sigma}_{\pi}$.
It is not hard to verify that $\vec{\sigma}_{\pi}$ 
implements $\pi$, and that there exists a probability measure $\mu$ such that
$(\vec{\sigma}_{\pi},\mu)$ is a sequential equilibrium of $G_{\corr}$. 
Defining $\mu$ is easy: the only information sets not reached with
positive probability (and hence $\mu$ is determined)
are the one where ``destroy" is played. At that point, the players'
play $\sigma$, so they are best responding to each other, no matter
what their beliefs are.

So now all we have to provide is a computational game $\G_{\corr}$ that represents
$G_{corr}$, where the games in $\G_{\corr}$ use cryptography instead
an envelope for the first part of the game.
Let $d$ be such
that $2^{d-1}\leq \ell <2^d$. Let $\G_{corr}$ be the sequence where
$G_n$ is the game where, at the empty history, player $1$ commits to a
$d$-bit string by using $d$ commitments in parallel, each with key
length $n-1$ and outputs the $d$ commitment strings as his action. Player
$2$ then plays a bitstring of length $d$ that can be viewed as a binary
representation of a number in $\{0, \ldots ,\ell-1 \}$. Player $1$
then plays a string that is intended to be the commitment keys of the
$d$ commitments. Then the players play a string representing their
action in $G$ (again using its binary representation). 
The utility are then given by the utility functions in $G$.

We now claim that $\G_{corr}$ represents $G_{corr}$. 

\begin{theorem}
\label{thm:Gcorr}
$\G_{corr}$ represents $G_{corr}$.
\end{theorem}
\begin{proof}
It is obvious that $\G_{corr}$ is a computable uniform sequence. We
now show that it represents $G_{corr}$. The mappings $\vec{f}$ of
histories maps player 1's commitments to a string $s$ to the action $s
\mod \ell$. 
(Notice that the fact we used $d$ commitments in parallel does not
change the fact that the commitments are perfectly binding and thus
this is well defined.)  
Actions of player $2$ are mapped to an action $s \mod \ell$ according
to their binary representation; if player 1 reveals $d$ valid keys in
$h$, then in $f_n(h)$ he plays ``open", and otherwise he plays
``destroy"; the
actions of $G$ are mapped in the obvious way.  

To show that UG4 holds, we proceed as follows: The mapping $\F$ for
a player $j$ other than $1$ and $2$ is obvious:
It is easy to compute using the TM $R$ of the commitment scheme if the
commitments were opened successfully or not, so $j$ can compute at which
information set of $G_{corr}$ he is at (given his view), and play the binary
representation of the action that the strategy plays at that
information set. 
For player 2, note that player $2$'s first action in $G_{\corr}$ can't
depend on player 1's action, 
since player 2's information set contains all the histories. Thus, a 
deterministic strategy $\sigma_2$ for player $2$ in $G_{\corr}$   
just plays an action in $\{0,\ldots,\ell-1\}$;   
$\F(\sigma_2)$ just plays the same action at player 2's first
information set in $\G_{\corr}$.  
Similarly to the other players,
$\F(\sigma_2)$ also plays the same action in $G$ as $\sigma_2$ when
player 2 is called upon to play again.
Given a deterministic strategy $\sigma_1$ for player 1,
if $\sigma_1$ plays $a$ at the first step in $\G_{\corr}$,  
$\F(\sigma_1)$ chooses uniformly at random one of the $d$-bit strings such that
$s = a \mod \ell$ (there are at most $2$ such strings), and plays the
commitments strings 
$C_1(1^n,s_1,r_1),\ldots,C_d(1^n,s_d,r_d)$, 
where $r=r_1||\ldots||r_d$ is the prefix of the random tape 
representing the randomness used to compute the commitment strings.
To play the action ``open'', it  computes $k_i= C_2(1^n,s_i,r_i)$ and
play $k_1||\ldots||k_d$; to play ``destroy'',  
it plays $k_1||\ldots||k_d\oplus1$ (a   
string other than the right keys). Again, it is obvious how player
$1$ plays in $G$. 
It is easy to see that $\F(\vec{\sigma})$ corresponds to
$\vec{\sigma}$, so UG4(a) holds. UG4(b) holds 
for all players trivially given these strategies. 

It is also obvious that UG4(c) holds for player $1$. Since the
information structure it faces at $\G_{corr}$ and $G_{corr}$ is 
essentially the same, anything it can do in $\G_{\corr}$ can be done
by a strategy in 
$G_{corr}$ by just looking at the distribution of actions in histories
that map to each information set.  

The other players have different information structures in
$\G_{\corr}$ and $G_{\corr}$, since they see the commitment strings in 
$\G_{corr}$.  We discuss UG4(c) for player 2 here; the argument in
the case of the others is similar (and simpler). Let $\sigma_i$ for
$i\neq 2$ be a strategy for player $i$ in $G_{corr}$, and let
$M_{i}=\F(\sigma_i)$. Let $M'$ be an arbitrary polynomial time
strategy for player $2$ in $\G_{corr}$, and let $D^n_1$ be the
distribution $M'$'s first action in $G_n$%
; let $D^n_{j,w}$ be the distribution over the actions of
$M'$ in $G$ given that the commitment was opened successfully, player
$1$ committed to $j$, and player $2$'s first move was $w$; and
let $D^n_w$ be the distribution over the actions of $M'$ in $G_n$ if
the commitment is not opened successfully and player $2$'s first move was $w$. 
Let $\sigma'_n$ be a strategy in $G_{\corr}$ for player $2$ that plays
according to these distributions. 
We claim that $\{\phi^{G_n}_{(M_1,M',\ldots,M_c)}\}_n$ is
indistinguishable from
$\{\rho^{G_{corr}}_{(\sigma_1,\sigma'_n,\ldots,\sigma_c)}\}_n$.  

Let $\phi^{G_n,1}_{(M_1,M',\ldots,M_c)}$ be the distribution over
histories ending at the first action of player $2$ when
$(M_1,M',\ldots,M_c)$ is played in $G_n$ and mapped using $f_n$ to
histories of $G_{corr}$, and let
$\rho^{G_{corr},1}_{(\sigma_1,\sigma'_n,\ldots,\sigma_c)}$ be the
distribution over partial histories ending at the first action of
player $2$ when $(\sigma_1,\sigma'_n,\ldots,\sigma_c)$ is played in
$G_{corr}$. We first claim that
$\{\phi^{G_n,1}_{(M_1,M',\ldots,M_c)}\}_n$ is indistinguishable from
$\{\rho^{G_{corr},1}_{(\sigma_1,\sigma'_n,\ldots,\sigma_c)}\}_n$. 
Assume, by way of contradiction, that it is not. 
This can happen only if, for
infinitely many $n$, $M'$ plays some action $a$ with 
probabilities that differ non-negligibly, depending on whether it is
faced with a commitment to different strings $s$ or $s'$.
But that means that for infinitely many $n$, it can distinguish those
two events with 
non-negligible probability. This contradicts the assumption that the
commitment scheme is secure. (Note that it is easy to show that, because
a single commitment has the hiding property, then even when $d$ such commitments are
run in parallel, no polynomial-time TM should be able to distinguish
between commitments to $s$ and $s'$.) 

It is easy to see that this also means that the distribution over
partial histories just before player $2$ plays again are also
indistinguishable. Now if the commitment is opened successfully, then
the information structure player $2$ faces in $\G_{\corr}$ is the
same as in $G_{\corr}$, 
and thus the statement is obviously true. If the commitments were not
opened, than by using a argument similar to that used for player 2's
first action, 
we can argue that if the distributions over partial histories just
after player $2$ plays again are not indistinguishable, then again we
can use that as a distinguisher for the commitment scheme. 
\end{proof}

By Theorems~\ref{seqeqrep} and \ref{thm:Gcorr}, since
$\vec{\sigma}_{\mu}$ (with the appropriate beliefs) is a sequential
equilibrium of $G_{\corr}$, $\F(\vec{\sigma}_{\mu})$ is a
computational sequential equilibrium of $\G_{\corr}$. 

\section{Conclusion}
The model introduced in this paper is a first step towards a better
understanding of polynomially bounded players playing finite games.
We defined a sense in which a sequence $\G$ of games represents a
single underlying game $G$, 
gave a novel definition of a computational sequential equilibrium, and
provided conditions under which a computational sequential equilibrium (and
hence also a computational NE) exists in $\G$.
Moreover, the model allows us to separate the cryptographic
analysis from the strategic analysis.
We show how we can use our model and definitions to analyze complex
cryptographic protocols in a way that captures our intuitions about
the rational behavior of the players in those protocols.

An important next step is to refine the model so it can capture more
complex cryptographic protocols. For example, some cryptographic
protocols do not have a unique map between histories and actions 
(e.g., a computationally binding commitment can map a string to
both  
$0$ or $1$ depending on the key). 
They also might have abstract actions
that are hard to compute (for instance, there might be strings that
are not valid commitments at all but it might be hard to compute
them), or require a few implementation steps to implement one abstract
step. One possible direction is to map a history \emph{and} the TMs'
views into histories in the game. While this might solve some
of the issues raised, it also introduces new challenges, which we intend
to investigate.

While in this paper we focus only on computationally bounded players
represented by polynomial-time TMs (which seems most appropriate for the
cryptographic applications we have in mind), we believe that our
approach of relating a sequence of games to a single underlying game,
and capturing computational indisitingushability with the information
structure of this game can be applied to other models of computations with the appropriate
adaptation of computational indisitingushability.

\bibliographystyle{plain}
\bibliography{z,joe}

\begin{thebibliography}{10}

\bibitem{DHR00}
Y.~Dodis, S.~Halevi, and T.~Rabin.
\newblock A cryptographic solution to a game theoretic problem.
\newblock In {\em CRYPTO 2000: 20th International Cryptology Conference}, pages
  112--130. Springer-Verlag, 2000.

\bibitem{goldreich01}
O.~Goldreich.
\newblock {\em Foundations of {C}ryptography, {V}ol. 1}.
\newblock Cambridge University Press, 2001.

\bibitem{goldwasser1984probabilistic}
S.~Goldwasser and S.~Micali.
\newblock Probabilistic encryption.
\newblock {\em Journal of Computer and System Sciences}, 28(2):270--299, 1984.

\bibitem{GLR13}
R.~Gradwohl, N.~Livne, and A.~Rosen.
\newblock Sequential rationality in cryptographic protocols.
\newblock {\em ACM Trans. Econ. Comput.}, 1(1):2:1--2:38, January 2013.

\bibitem{HPR15}
J.~Y. Halpern, R.~Pass, and D.~Reichman.
\newblock On the nonexistence of equilibrium in computational games.
\newblock 2015.

\bibitem{HPS14b}
J.~Y. Halpern, R.~Pass, and L.~Seeman.
\newblock Not just an empty threat: subgame-perfect equilibrium in repeated
  games played by computationally bounded players.
\newblock In {\em Proc.~WINE 2014: 10th Conference on Web and Internet
  Economics}, pages 249--262, 2014.

\bibitem{HPS14}
J.~Y. Halpern, R.~Pass, and L.~Seeman.
\newblock The truth behind the myth of the folk theorem.
\newblock In {\em Proc.~5th Conference on Innovations in Theoretical Computer
  Science (ITCS '14)}, pages 543--554, 2014.

\bibitem{hubavcek2013limits}
P.~Hub{\'a}{\v{c}}ek, J.~B. Nielsen, and A.~Rosen.
\newblock Limits on the power of cryptographic cheap talk.
\newblock In {\em Advances in Cryptology--CRYPTO 2013}, pages 277--297.
  Springer, 2013.

\bibitem{hubavcek2014cryptographically}
P.~Hub{\'a}{\v{c}}ek and S.~Park.
\newblock Cryptographically blinded games: leveraging players' limitations for
  equilibria and profit.
\newblock In {\em Proc.~15th ACM Conference on Economics and Computation},
  pages 207--208, 2014.

\bibitem{Katz08}
Jonathan Katz.
\newblock Bridging game theory and cryptography: Recent results and future
  directions.
\newblock In {\em Theory of Cryptography}, pages 251--272. 2008.

\bibitem{KW82}
D.~M. Kreps and R.~B. Wilson.
\newblock Sequential equilibria.
\newblock {\em Econometrica}, 50:863--894, 1982.

\bibitem{OR94}
M.~J. Osborne and A.~Rubinstein.
\newblock {\em A Course in Game Theory}.
\newblock MIT Press, Cambridge, Mass., 1994.

\bibitem{Selten65}
R.~Selten.
\newblock Spieltheoretische behandlung eines oligopolmodells mit
  nachfragetr\"{a}gheit.
\newblock {\em Zeitschrift f\"{u}r Gesamte Staatswissenschaft}, 121:301--324
  and 667--689, 1965.

\bibitem{Simon55}
H.~A. Simon.
\newblock A behavioral model of rational choice.
\newblock {\em Quarterly Journal of Economics}, 49:99--118, 1955.

\bibitem{UV04}
A.~Urbano and J.~E. Vila.
\newblock Computationally restricted unmediated talk under incomplete
  information.
\newblock {\em Economic Theory}, 23(2):283--320, 2004.

\end{thebibliography}


\end{document}